\documentclass[a4paper,UKenglish,cleveref, autoref]{lipics-v2019}
\usepackage[utf8]{inputenc}
\usepackage[T1]{fontenc}
\usepackage{lmodern}

\usepackage{amssymb,amsmath}
\usepackage{graphicx}
\usepackage{enumerate}
\usepackage{paralist}
\usepackage{tikz}
\usetikzlibrary{shapes}

\usepackage{xcolor}
\usepackage{mathtools}
\usepackage{microtype}
\usepackage{amsfonts}
\usepackage{comment}
\usepackage{mathrsfs}
\usepackage[ruled,vlined]{algorithm2e}
\usepackage{blindtext}

\usepackage[absolute]{textpos}
\usepackage{enumitem}
\usepackage{todonotes}
\usepackage{longtable}

\definecolor{blue}{rgb}{0.1,0.2,0.5}
\definecolor{brown}{rgb}{0.6,0.6,0.2}
\usepackage{cleveref}
\usepackage{enumerate}
\usepackage{latexsym}

\newcommand{\Oh}{\mathcal{O}}
\newcommand{\Ohtilde}{\tilde{\Oh}}
\newcommand{\Cc}{\mathcal{C}}
\newcommand{\Dd}{\mathcal{D}}

\newcommand{\tw}{\mathrm{tw}}
\newcommand{\eps}{\varepsilon}

\newcommand{\N}{\mathbb{N}}

\renewcommand{\phi}{\varphi}
\renewcommand{\epsilon}{\varepsilon}

\renewcommand{\leq}{\leqslant}
\renewcommand{\geq}{\geqslant}

\newcommand{\FVS}{\textsc{Feedback Vertex Set}\xspace}
\newcommand{\MIF}{\textsc{Max Induced Forest}\xspace}

\title{Subexponential-time algorithms for finding large induced sparse subgraphs} 

\author{Jana Novotn\'a}{Department of Applied Mathematics, Faculty of Mathematics and Physics, Charles University, Prague, Czech Republic}{janca@kam.mff.cuni.cz}{}{Supported by student grants GAUK 1277018, SVV-2017-260452.}

\author{Karolina Okrasa}{Faculty of Mathematics and Information Science, Warsaw University of Technology, Poland}{k.okrasa@mini.pw.edu.pl}{}{}

\author{Micha\l{} Pilipczuk}{Institute of Informatics, Faculty of Mathematics, Informatics and Mechanics, University of Warsaw, Poland}{michal.pilipczuk@mimuw.edu.pl}{}{This work is 
a part of project TOTAL that has received funding from the European Research Council (ERC) 
under the European Union's Horizon 2020 research and innovation programme (grant agreement No.~677651).}

\author{Pawe\l{} Rz\k{a}\.zewski}{Faculty of Mathematics and Information Science, Warsaw University of Technology, Poland}{p.rzazewski@mini.pw.edu.pl}{0000-0001-7696-3848}{Partially supported by Polish National Science Centre grant no.\ 2018/31/D/ST6/00062.}

\author{Erik Jan van Leeuwen}{Department of Information and Computing Sciences, Utrecht University, The Netherlands}{e.j.vanleeuwen@uu.nl}{}{}

\author{Bartosz Walczak}{Department of Theoretical Computer Science, Faculty of Mathematics and Computer Science, Jagiellonian University, Kraków, Poland}{walczak@tcs.uj.edu.pl}{}{Partially supported by Polish National Science Centre grant no.\ 2015/17/B/ST6/01873.}

\authorrunning{J. Novotn\'a, K. Okrasa, Mi.\ Pilipczuk, P. Rz\k{a}\.zewski, E. J. van Leeuwen, and B. Walczak}

\Copyright{Jana Novotn\'a, Karolina Okrasa, Micha\l{} Pilipczuk, Pawe\l{} Rz\k{a}\.zewski, Erik Jan van Leeuwen, and Bartosz Walczak}



\category{}

\relatedversion{}

\supplement{}

\acknowledgements{The results presented in this paper were obtained during the Parameterized Algorithms Retreat of the algorithms group of the University of Warsaw (PARUW), held in Karpacz in February 2019.
This Retreat was financed by the project CUTACOMBS, which has received funding from the European Research Council (ERC) 
under the European Union's Horizon 2020 research and innovation programme (grant agreement No.~714704).}

\nolinenumbers 

\hideLIPIcs  

\EventEditors{John Q. Open and Joan R. Access}
\EventNoEds{2}
\EventLongTitle{42nd Conference on Very Important Topics (CVIT 2016)}
\EventShortTitle{CVIT 2016}
\EventAcronym{CVIT}
\EventYear{2016}
\EventDate{December 24--27, 2016}
\EventLocation{Little Whinging, United Kingdom}
\EventLogo{}
\SeriesVolume{42}
\ArticleNo{23}

\begin{document}

\maketitle

\vskip -0.5cm
\begin{picture}(0,0)
\put(402,-300)
{\hbox{\includegraphics[width=40px]{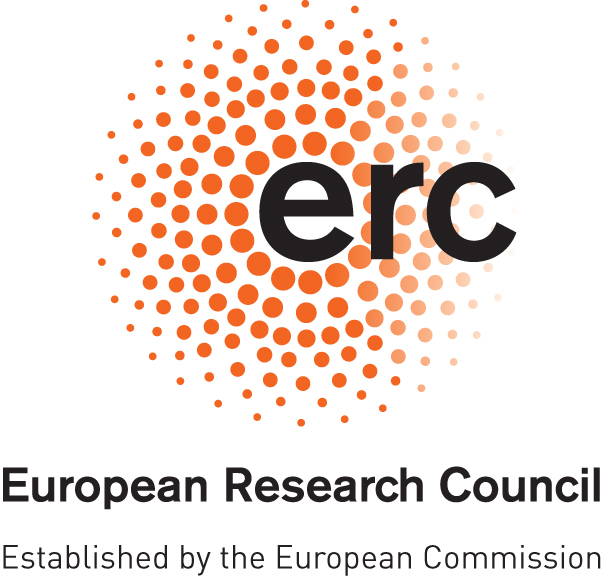}}}
\put(392,-360)
{\hbox{\includegraphics[width=60px]{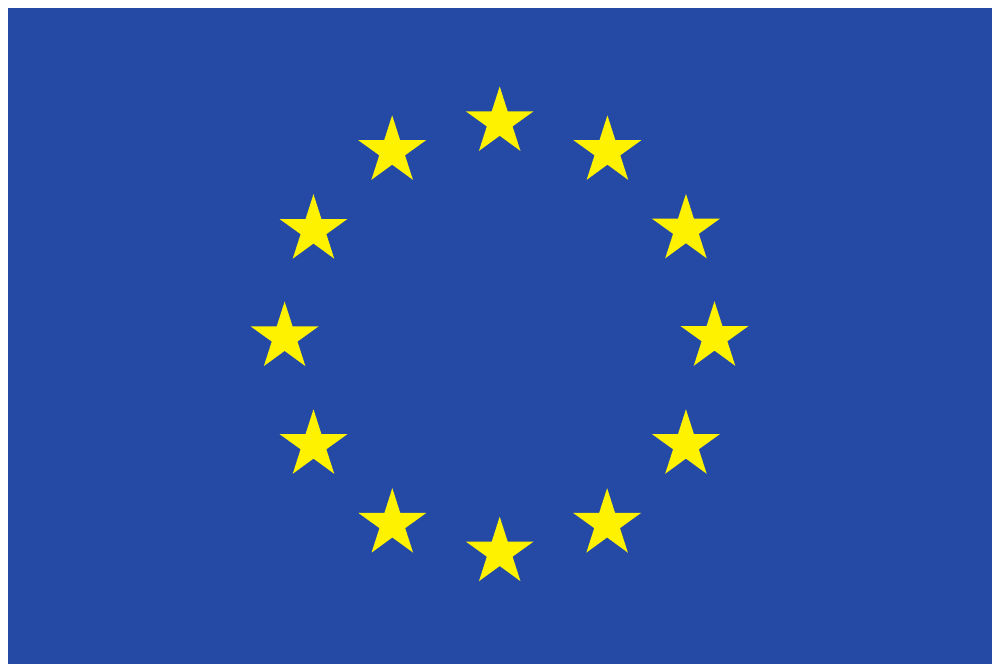}}}
\end{picture}

\begin{abstract}
Let $\mathcal{C}$ and $\mathcal{D}$ be hereditary graph classes.
Consider the following problem: given a graph $G\in\mathcal{D}$, find a largest, in terms of the number of vertices, induced subgraph of $G$ that belongs to $\mathcal{C}$.
We prove that it can be solved in $2^{o(n)}$ time, where $n$ is the number of vertices of $G$, if the following conditions are satisfied:
\begin{compactitem}
\item the graphs in $\mathcal{C}$ are sparse, i.e., they have linearly many edges in terms of the number of vertices;
\item the graphs in $\mathcal{D}$ admit balanced separators of size governed by their density, e.g., $\mathcal{O}(\Delta)$ or $\mathcal{O}(\sqrt{m})$, where $\Delta$ and $m$ denote the maximum degree and the number of edges, respectively; and
\item the considered problem admits a single-exponential fixed-parameter algorithm when parameterized by the treewidth of the input graph.
\end{compactitem}
This leads, for example, to the following corollaries for specific classes $\mathcal{C}$ and $\mathcal{D}$:
\begin{compactitem}
\item a largest induced forest in a $P_t$-free graph can be found in $2^{\tilde{\mathcal{O}}(n^{2/3})}$ time, for every fixed $t$; and
\item a largest induced planar graph in a string graph can be found in $2^{\tilde{\mathcal{O}}(n^{3/4})}$ time.
\end{compactitem}
\end{abstract}

\section{Introduction}\label{sec:intro}

Many optimization problems in graphs can be expressed as follows: given a graph $G$, find a largest vertex set $A$ such that $G[A]$, the subgraph of $G$ induced by $A$, satisfies some~property.
Examples include {\sc{Independent Set}} (the property of being edgeless), {\sc{Feedback Vertex Set}} (the property of being acyclic), and {\sc{Planarization}} (the property of being planar). 
Here, {\sc{Feedback Vertex Set}} and {\sc{Planarization}} are customarily phrased in the complementary form that asks for minimizing the complement of $A$: 
given $G$, find a smallest vertex set $X$ such that $G-X$ has the desired property. 
While all problems considered in this paper can be viewed in these two ways, for the sake of clarity we focus on the maximization formulation.

Formally, we shall consider the following {\sc{Max Induced $\Cc$-Subgraph}} problem.
Fix a graph class $\Cc$ that is {\em{hereditary}}, that is, closed under taking induced subgraphs. Then, given a graph $G$, the goal is to find a largest vertex subset
$A$ such that $G[A]\in \Cc$. Our focus is on exact algorithms for this problem with running time expressed in terms of $n$, the number of vertices of $G$.
Clearly, as long as the graphs from $\Cc$ can be recognized in polynomial time, the problem can be solved in $2^n\cdot n^{\Oh(1)}$ time by brute-force; 
we are interested in non-trivial improvements over this approach.

The complexity of {\sc{Max Induced $\Cc$-Subgraph}} was studied as early as in 1980 by Lewis and Yannakakis~\cite{LewisY80}, who proved that when the graph class $\Cc$ does not contain all graphs, the 
problem is NP-hard. Recently, Komusiewicz~\cite{Komusiewicz18} inspected the reduction of Lewis and Yannakakis and concluded that under the Exponential Time Hypothesis (ETH) one can even exclude the existence of 
{\em{subexponential-time}} algorithms for the problem, that is, ones with running time $2^{o(n)}$.
While the result of Komusiewicz~\cite{Komusiewicz18} excludes significant improvements in the running time, there is still room for improvement in the base of the exponent. 
Indeed, for various classes of graphs $\Cc$, algorithms with running time $\Oh((2-\eps)^n)$ for some $\eps>0$ are known; 
see e.g.~\cite{BliznetsFPV16,FominGLS16,FominTV11,FominTV15,PilipczukP12} and the references therein.

Another direction, which is of main interest to us, is to impose more conditions on the input graphs $G$ in the hope of obtaining faster algorithms for restricted cases.
Formally, we fix another hereditary graph class $\Dd$ and consider {\sc{Max Induced $\Cc$-Subgraph}} where the input graph $G$ is additionally required to belong to $\Dd$.

In this line of research, the class $\Cc$ of edgeless graphs, which corresponds to the classical {\sc{Max Independent Set}} ({\sc{MIS}}) problem, has been extensively studied.
Suppose $\Dd$ is the class of {\em{$H$-free graphs}}, that is, graphs that exclude some fixed graph $H$ as an induced subgraph.
As observed by Alekseev~\cite{Alekseev82}, the problem is NP-hard on $H$-free graphs unless $H$ is a path or a subdivision of the claw ($K_{1,3}$); 
the reduction of~\cite{Alekseev82} actually excludes the~existence of a subexponential-time algorithm under ETH in these cases.
On the positive side, the maximal classes for which polynomial-time algorithms are known are the $P_6$-free graphs~\cite{GrzesikKPP19} and the fork-free graphs~\cite{LozinM08}.
It would be consistent with our knowledge if {\sc{MIS}} was polynomial-time solvable on $H$-free graphs whenever $H$ is a path or a subdivision of the~claw.

It turns out that if we only aim at subexponential-time instead of polynomial-time algorithms, many more tractability results can be obtained for {\sc{MIS}}, 
and usually they are also much simpler conceptually. Bacs\'o et al.~\cite{BacsoLMPTL19} showed that {\sc{MIS}} can be solved in $2^{\Oh(\sqrt{tn\log n})}$ time on $P_t$-free graphs, for every $t\in \N$.
Very recently, Chudnovsky et al.~\cite{ChudnPPT_2019} reported a $2^{\Oh(\sqrt{n\log n})}$-time algorithm on {\em{long-hole-free graphs}}, 
which are graphs that exclude every cycle of length at least $5$ as an induced subgraph.

In the light of the results above, it is natural to ask whether structural assumptions on the class $\Dd$ from which the input is drawn, like e.g.\ $P_t$-freeness,
can help in the design of subexponential-time algorithms for other maximum induced subgraph problems, beyond $\Cc$ being the class of edgeless graphs.
This is precisely the question we investigate in this work.

\subparagraph*{Our contribution.} 
We identify three properties that together provide a way to solve the {\sc{Max Induced $\Cc$-Subgraph}} problem on graphs from $\Dd$ in subexponential time, where $\Cc$ and $\Dd$ are hereditary graph classes. 
They are as follows:
\begin{itemize}
\item The class $\Cc$ should consist of {\em{sparse}} graphs. To be specific, let us assume that every $n$-vertex graph from $\Cc$ has $\Oh(n)$ edges.
\item The class $\Dd$ may contain dense graphs, but they should admit balanced separators whose size is somehow governed by the density. To be specific, let us assume that every graph from $\Dd$ 
with maximum degree $\Delta$ has a balanced separator of size $\Oh(\Delta)$, or that every graph from $\Dd$ with $m$ edges has a balanced separator of size $\Oh(\sqrt{m})$.
\item The {\sc{Max Induced $\Cc$-Subgraph}} problem on graphs from $\Dd$ can be solved in $2^{\Ohtilde(w)}\cdot\nobreak n^{\Oh(1)}$ time, where $w$ is the treewidth of the input graph.
Here, notation $\Ohtilde(\cdot)$ hides polylogarithmic factors.
\end{itemize}
We show that if these conditions are simultaneously satisfied, then the {\sc{Max Induced $\Cc$-Subgraph}} problem on graphs from $\Dd$ can be solved in $2^{\Ohtilde(n^{2/3})}$ time in the presence of balanced separators
of size $\Oh(\Delta)$ and in $2^{\Ohtilde(n^{3/4})}$ time for balanced separators of size $\Oh(\sqrt{m})$. The precise statement and proof of this result can be found in \cref{sec:main}.

The conditions on $\Cc$ look natural and are satisfied by various specific classes of interest, like forests (corresponding to {\sc{Feedback Vertex Set}}) and planar graphs (corresponding to {\sc{Planarization}}).
On the other hand, the condition on $\Dd$ looks more puzzling. However, there are certain non-sparse classes of graphs where the existence of such balanced separators has been established.
For instance, balanced separators of size $\Oh(\Delta)$ are known to exist in $P_t$-free graphs for any fixed $t\in \N$~\cite{BacsoLMPTL19}, and in long-hole-free graphs~\cite{ChudnPPT_2019}.
The existence of balanced separators of size $\Oh(\sqrt{m})$ is known for {\em{string graphs}}, which are intersection graphs of arc-connected subsets of the plane, 
and more generally for intersection graphs of connected subgraphs in any proper minor-closed class~\cite{DBLP:conf/innovations/Lee17}. 
All these observations yield a number of concrete corollaries to our main result, which are gathered in \cref{sec:cors}.
In \cref{sec:lbs}, we discuss some lower bounds: we show that if $\Cc$ is the class of forests (corresponding to the {\sc{Feedback Vertex Set}} problem) and $\Dd$ is characterized by a single excluded induced subgraph,
then under the Exponential Time Hypothesis one cannot hope for subexponential-time algorithms in greater generality than provided by our main result.

\section{Main result}\label{sec:main}

We use standard graph notation. We assume the reader's familiarity with treewidth. We recall some notation for tree decompositions in Section~\ref{sec:deg-dp}, where it is actually needed.

For a graph $G$, a set $S\subseteq V(G)$ is a \emph{balanced separator} if every connected component of $G-S$ has at most $\frac{2}{3}|V(G)|$ vertices.
It is known that small balanced separators can be used to construct tree decompositions of small width, as made explicit in the following lemma.

\begin{lemma}[\cite{DVORAK2019137}]\label{lem:sn}
If every subgraph of a graph\/ $G$ has a balanced separator of size at most\/ $k$, then the treewidth of\/ $G$ is\/ $\Oh(k)$.
\end{lemma}

Now, we are ready to state and prove our main result.

\begin{theorem} \label{thm:main}
Let\/ $\Cc$ and\/ $\Dd$ be classes of graphs that satisfy the following conditions:
\begin{enumerate}[label=(P\arabic*),ref=(P\arabic*),leftmargin=*]
\item Every\/ $n$-vertex graph from\/ $\Cc$ has\/ $\Oh(n)$ edges. \label{prop:sparse}
\item The class\/ $\Dd$ is closed under taking induced subgraphs. \label{prop:hereditary}
\item Given a graph\/ $G\in \Dd$ with\/ $n$ vertices and treewidth\/ $w$, one can find a largest set\/ $A \subseteq V(G)$ such that\/ $G[A]\in \Cc$ in\/ $2^{\Ohtilde(w)}\cdot n^{\Oh(1)}$ time. \label{prop:ctw}
\end{enumerate}
\noindent Furthermore, let the class\/ $\Dd$ satisfy one of the following conditions:
\begin{enumerate}[label=(P4\alph*),ref=(P4\alph*),widest=a,leftmargin=*]
\item Every graph in\/ $\Dd$ with maximum degree\/ $\Delta$ has a balanced separator of size\/ $\Oh(\Delta)$, or \label{prop:separatorDelta}
\item Every graph in\/ $\Dd$ with\/ $n$ vertices and maximum degree\/ $\Delta$ has a balanced separator of size\/ $\Oh(\sqrt{n\Delta})$. \label{prop:separatorM} 
\end{enumerate}
\noindent Then, given an\/ $n$-vertex graph\/ $G\in \Dd$, one can find a largest set\/ $A\subseteq V(G)$ such that\/ $G[A]\in \Cc$ in time
\begin{enumerate}[label=(\alph*),ref=(\alph*),widest=a,leftmargin=*]
\item $2^{\Ohtilde(n^{2/3})}$, if\/ $\Dd$ satisfies \ref{prop:separatorDelta}, or
\item $2^{\Ohtilde(n^{3/4})}$, if\/ $\Dd$ satisfies \ref{prop:separatorM}.
\end{enumerate}
\end{theorem}

\begin{proof}
Let a constant $\tau$ be defined as follows, depending on which of the two conditions is satisfied by $\Dd$:
\[
\tau = \begin{cases}
{1/3} & \text{ if $\Dd$ satisfies \ref{prop:separatorDelta},}\\
{1/4} & \text{ if $\Dd$ satisfies \ref{prop:separatorM}.}
\end{cases} 
\]
We devise a branching algorithm that finds a largest set $A\subseteq V(G)$ such that $G[A]\in \Cc$ in $2^{\Ohtilde(n^{1-\tau})}$ time.
This matches the complexity bounds from the statement of the theorem.

Let $G\in \Dd$ be the input graph and $n$ be the number of its vertices. Consider a fixed solution $A$, that is, a largest set $A\subseteq V(G)$ such that $G[A]\in \Cc$.
Let $A' \subseteq A$ be the set of vertices of degree greater than $n^\tau$ in $G[A]$. By property \ref{prop:sparse}, we have $|A'| = \Oh(n/n^\tau) = \Oh(n^{1-\tau})$.

The algorithm guesses the set $A'$ exhaustively, by trying all subsets of $V(G)$ of the appropriate sizes $\Oh(n^{1-\tau})$, which results in $n^{\Oh(n^{1-\tau})} = 2^{\Ohtilde(n^{1-\tau})}$ branches.
Fix one such branch and assume, for the purpose of further description of the algorithm, that it corresponds to the true set $A'$ (i.e., the one obtained from the fixed solution $A$).
Let $G'=G-A'$.

Suppose that $G'$ contains a vertex $v$ of degree at least $n^{2\tau}$. 
If $v \in A$, then $v$ has degree at most $n^\tau$ in $G[A]$ (since $v \notin A'$). 
The algorithm further guesses that $v\notin A$ and discards $v$ (one branch), or it guesses that $v\in A$ and discards all but at most $n^\tau$ neighbors of $v$ in $G'$ (at most $n^{n^\tau}$ branches).
In the latter case, we do \emph{not} fix the assumption that $v$ or any particular neighbor of $v$ belongs to $A$, so that the vertices that have survived this step can still be discarded in subsequent branching steps.

The step described above is repeated exhaustively. The overall number of branches generated in this way can be bounded as follows, where $k=|V(G')|$:
\begin{align*}
F(k) & \leq  F(k-1) + n^{n^\tau}\cdot F(k-(n^{2\tau} - n^\tau))\\
     & \leq  F(k-2) + n^{n^\tau}\cdot F(k-(n^{2\tau} - n^\tau)) + n^{n^\tau}\cdot F(k-(n^{2\tau} - n^\tau))\\
     & \leq  \ldots \leq F(k-(n^{2\tau} - n^\tau))+(n^{2\tau} - n^\tau) \cdot n^{n^\tau}\cdot F(k-(n^{2\tau} - n^\tau))\\
     & =  (n^{2\tau} - n^\tau+1) \cdot n^{n^\tau} \cdot F(k-(n^{2\tau} - n^\tau))  \\
     & \leq  \left ( (n^{2\tau} - n^\tau+1) \cdot n^{n^\tau} \right ) ^{k/(n^{2\tau} - n^\tau)}\\
     & \leq  \left ( (n^{2\tau} - n^\tau+1) \cdot n^{n^\tau} \right ) ^{n/(n^{2\tau} - n^\tau)} = n^{\Oh(n^{1+\tau - {2\tau}})}=2^{\Ohtilde(n^{1-\tau})}.
\end{align*}
Once the branching step can no longer be applied, we obtain an induced subgraph $G''$ of $G'$ of maximum degree less than $n^{2\tau}$.
In the branch where all the choices have been made correctly (i.e., according to the fixed solution $A$), $G''$ still contains all vertices from $A \setminus A'$.

By property \ref{prop:hereditary}, we have $G'' \in \Dd$. Thus $G''$ satisfies either \ref{prop:separatorDelta} or \ref{prop:separatorM}, which means that $G''$ has a balanced separator of size $\Oh(n^{2/3})$ in the former case or $\Oh(\sqrt{n \cdot n^{1/2}}) = \Oh(n^{3/4})$ in the latter case.
In both cases, the size of the separator is $\Oh(n^{1-\tau})$.
Moreover, by the same argument, balanced separators of that size also exist in every subgraph of $G''$. 
Therefore, by \cref{lem:sn}, we conclude that $G''$ has treewidth $\Oh(n^{1-\tau})$. Since $|A'|\leq \Oh(n^{1-\tau})$, it follows that the graph $G[V(G'') \cup A']$ also has treewidth $\Oh(n^{1-\tau})$.


We know that $G[V(G'') \cup A'] \in \Dd$ and, in the branch where all choices have been made correctly, this graph contains the entire maximum-size solution $A$. 
Now, we apply the procedure assumed in \ref{prop:ctw} to the graph $G[V(G'') \cup A']$ and observe that in the correct branch it finds some maximum-size solution (possibly different from $A$).
Let us point out that in this step it is not sufficient to consider only the graph $G''$, as the vertices from $A'$ introduce some additional constraints on the solution we are looking for.

For the time complexity, the algorithm considers $2^{\Ohtilde(n^{1-\tau})}$ branches and in each of them it executes the procedure assumed in \ref{prop:ctw} in $2^{\Ohtilde(n^{1-\tau})}$ time,
which gives the total running time of $2^{\Ohtilde(n^{1-\tau})}$.
\end{proof}

\begin{remark}\label{rem:generalization}
The condition \ref{prop:sparse} in the statement of \cref{thm:main} can be relaxed to ``every $n$-vertex graph from $\Cc$ has $\Oh(n^{2-\epsilon})$ edges, for some constant $\epsilon > 0$''.
Then, we can follow the same approach with the following modification: we choose $\tau=1-\frac{2}{3}\epsilon$ in case of~\ref{prop:separatorDelta} and $\tau=1-\frac{3}{4}\epsilon$ in case of~\ref{prop:separatorM}, and replace the threshold for branching on high-degree vertices from $n^{2\tau}$ to $n^{2\tau+\epsilon-1}$.
This way, we obtain algorithms with running time $2^{\Ohtilde(n^{1-\eps/3})}$ for property~\ref{prop:separatorDelta} and $2^{\Ohtilde(n^{1-\eps/4})}$ for property~\ref{prop:separatorM}.
This running time is subexponential for every $\eps>0$.

One can also imagine unifying properties \ref{prop:separatorDelta} and  \ref{prop:separatorM} into the existence of a balanced separator of size $\Oh(n^{\alpha}\Delta^\beta)$, for some constants $\alpha,\beta$.
However, then, one needs to be careful when choosing $\tau$ so that it belongs to the interval $[0,1]$. As we did not find concrete examples of interesting graph classes $\Dd$ for which this approach would yield non-trivial results and which would not satisfy either \ref{prop:separatorDelta} or \ref{prop:separatorM}, we refrain from discussing further details here.
\end{remark}


\section{Corollaries}\label{sec:cors}

In this section, we discuss possible classes $\Cc$ and $\Dd$ which satisfy the conditions of \cref{thm:main}.
For some choices of $\Cc$, we obtain well-studied computational problems:
\begin{enumerate}
\item for matchings, we obtain \textsc{Max Induced Matching},
\item for forests, we obtain \MIF, also known as \FVS,
\item for graphs of maximum degree $d$, where $d$ is fixed, we obtain \textsc{Max Induced Degree-$d$ Subgraph},
\item for planar graphs, we obtain \textsc{Max Induced Planar Subgraph}, also known as \textsc{Planarization},
\item for graphs embeddable in $\Sigma$, where the surface $\Sigma$ is fixed, we obtain \textsc{Max Induced $\Sigma$-Embeddable Subgraph},
\item for graphs of degeneracy at most $d$, where $d$ is fixed, we obtain \textsc{Max Induced $d$-Degenerate Subgraph}.
\end{enumerate}
It is clear that all these classes satisfy property~\ref{prop:sparse} of \cref{thm:main}.

Given a graph of treewidth $w$, its tree decomposition of width at most $4w+3$ can be computed in $2^{\Oh(w)}\cdot n^{2}$ time (see e.g.~\cite[Section 7.6]{platypus}). 
Therefore, for the purpose of verifying property~\ref{prop:ctw}, we can assume that a tree decomposition of width $\Oh(w)$ is additionally provided on input.
While $2^{\Ohtilde(w)}\cdot n^{\Oh(1)}$-time algorithms are quite straightforward and well known for the first two problems on the list, this is not necessarily the case for the others.
For \textsc{Max Induced Degree-$d$ Subgraph}, an algorithm with running time $2^{\Oh(w)}\cdot n^{\Oh(1)}$ can be easily derived from the meta-theorem of Pilipczuk~\cite{10.1007/978-3-642-22993-0_47}.
Algorithms for \textsc{Max Induced Planar Subgraph} and, more generally, \textsc{Max Induced $\Sigma$-Embeddable Subgraph}, were provided by Kociumaka and Pilipczuk~\cite{DBLP:journals/corr/KociumakaP17}. 
Finally, we give a suitable algorithm for \textsc{Max Induced $d$-Degenerate Subgraph} in \cref{lem:deg-dp} in \cref{sec:deg-dp}.

It may be tempting to consider, as $\Cc$, the graphs with no even cycle $C_{2k}$ (not necessarily induced), for some fixed integer $k\geq 2$.
This is because such graphs have $\Oh(n^{2-\Omega(1/k)})$ edges~\cite{BONDY197497}, and thus they satisfy the generalization of property \ref{prop:sparse} mentioned in \cref{rem:generalization} for $\epsilon=\Omega(1/k)$.
However, for these classes, property \ref{prop:ctw} turns out to be problematic: 
for any fixed $\ell \geq 5$, there is no algorithm for a minimum set of vertices hitting all (non-induced) copies of $C_\ell$ in a graph with treewidth $w$ with running time $2^{o(w^2)}\cdot n^{\Oh(1)}$ unless the ETH fails~\cite{10.1007/978-3-642-22993-0_47} (this bound appears to be essentially tight, as the problem can be solved in $2^{\Ohtilde(w^2)}\cdot n^{\Oh(1)}$  time~\cite{DBLP:journals/iandc/CyganMPP17}). It is unclear whether the additional assumption that the input graph belongs to some class $\Dd$, considered here, can help.


Now, let us consider classes $\Dd$. Examples of classes satisfying property \ref{prop:separatorDelta} in \cref{thm:main} come from forbidding some induced subgraphs.  Bacsó et al.~\cite{BacsoLMPTL19} proved that $P_t$-free graphs with maximum degree $\Delta$ have treewidth $\Oh(\Delta \cdot t)$. Very recently, Chudnovsky et al.~\cite{ChudnPPT_2019} observed that \emph{long-hole-free graphs}, that is, graphs with no induced cycles of length at least $5$, also have balanced separators of size $\Oh(\Delta)$. 

An example of a class satisfying property \ref{prop:separatorM} is the class of string graphs---intersection graphs of arc-connected subsets of the plane. Lee~\cite{DBLP:conf/innovations/Lee17} showed that they admit balanced separators of size $\Oh(\sqrt{m})$, where $m$ is the number of edges. In fact, he proved a more general result that if ${\cal M}$ is a class of graphs excluding a fixed graph as a minor, then intersection graphs of connected subgraphs of graphs from ${\cal M}$ admit balanced separators of size $\Oh(\sqrt{m})$. 
String graphs are precisely the intersection graphs of connected subgraphs of planar graphs. 

Summing up, we obtain the following. 

\begin{corollary} \label{cor:ptfree}
Each of the following problems can be solved in\/ $2^{\Ohtilde(n^{2/3})}$ time on\/ $P_t$-free graphs (for every fixed\/ $t$) and in long-hole-free graphs, and
in\/ $2^{\Ohtilde(n^{3/4})}$ time on string graphs:
\begin{enumerate}
\item \textsc{Max Induced Matching},
\item \MIF,
\item \textsc{Max Induced Degree-$d$ Subgraph}, for every fixed\/ $d\in \N$,
\item \textsc{Max Induced Planar Subgraph},
\item \textsc{Max Induced $\Sigma$-Embeddable Subgraph}, for every fixed surface\/ $\Sigma$,
\item \textsc{Max Induced $d$-Degenerate Subgraph}, for every fixed\/ $d\in \N$.
\end{enumerate}
\end{corollary}

We note that subexponential-time algorithms for \textsc{Max Induced Matching} and \MIF on string graphs were already known~\cite{DBLP:journals/algorithmica/BonnetR19}, even with a better running time than provided above.
As we have argued, in \cref{cor:ptfree}, we can replace string graphs with intersection graphs of connected subgraphs of graphs from ${\cal M}$, where ${\cal M}$ is any class of graphs excluding a fixed graph as a minor; this is because the result of Lee~\cite{DBLP:conf/innovations/Lee17} holds in that generality.


%

\section{\MIF in $H$-free graphs}\label{sec:lbs}

Our original motivation was the \MIF problem.
In the previous section, we discussed a subexponential-time algorithm solving it on $P_t$-free graphs.
We now show that as long as the considered class of inputs $\Dd$ is characterized by a single excluded induced subgraph, that is,
we investigate \MIF on $H$-free graphs for a fixed graph $H$, we cannot hope for more positive results.
Namely, it turns out that if $H$ is not a linear forest (i.e., a collection of vertex-disjoint paths), the problem is unlikely to admit a polynomial-time or even a subexponential-time algorithm on $H$-free graphs. 
Specifically, we obtain the following dichotomy.

\begin{theorem}\label{thm:fvs-neg}
Let\/ $H$ be a fixed graph.
\begin{enumerate}
\item If\/ $H$ is a linear forest, then the \MIF problem can be solved in\/ $2^{\Ohtilde(n^{2/3})}$ time on\/ $H$-free graphs with\/ $n$ vertices. \label{item:algo}
\item Otherwise, on\/ $H$-free graphs, the \MIF problem is NP-complete and cannot be solved in\/ $2^{o(n)}$ time unless the ETH fails. \label{item:hardness}
\end{enumerate}
\end{theorem}

\cref{thm:fvs-neg} \ref{item:algo} follows from \cref{cor:ptfree}, because every linear forest is an induced subgraph of some path. 
Statement \ref{item:hardness} follows from a combination of arguments already existing in the literature. 
However, since the proof is simple, we include it for the sake of completeness.

We prove \autoref{thm:fvs-neg} \ref{item:hardness} in two steps. First, we consider graphs $H$ that contain a cycle or two branch vertices, that is, vertices of degree at least $3$.
In this case, we can apply the standard argument of subdividing every edge a suitable number of times, cf.~\cite[Theorem~3]{DBLP:journals/tcs/ChiarelliHJMP18}.

\begin{lemma}\label{lem:fvs-subdivide}
Let\/ $H$ be a fixed graph that either contains a cycle or has a connected component with at least two branch vertices.
Then \MIF is NP-complete on\/ $H$-free graphs. Moreover, there is no algorithm solving \MIF in\/ $2^{o(n)}$ time for\/ $n$-vertex\/ $H$-free graphs unless the ETH fails.
\end{lemma}

\begin{proof}
We reduce from \MIF in graphs with maximum degree $6$; it is known that this problem is NP-complete and has no subexponential-time algorithm assuming ETH~\cite{platypus}. Let $G$ be a graph with $n$ vertices and maximum degree $6$. Let $G^*$ be the graph obtained from $G$ by subdividing every edge $|V(H)|+1$ times. It is straightforward to observe that $G$ has an induced forest on $n-k$ vertices if and only if $G^*$ has an induced forest on $|G^*|-k$ vertices. Moreover, the number of vertices in $G^*$ is linear in $n$.

Finally, we show that $G^*$ is $H$-free. First, observe that if $H$ contains a cycle, then $H$ cannot be a subgraph of $G^*$, as the girth of $G^*$ is greater than $|V(H)|+1$. On the other hand, the distance between any two branch vertices in $G^*$ is at least $|V(H)|+1$, so $G^*$ does not contain $H$ as a subgraph in case $H$ has two branch vertices in the same connected component.
\end{proof}

By \cref{lem:fvs-subdivide}, the only graphs $H$ for which we might hope for a polynomial-time or even a subexponential-time algorithm for \MIF on $H$-free graphs are collections of disjoint subdivided stars. To resolve this case, we will show that the problem remains hard for line graphs. Recall that the line graph $L(G)$ of a graph $G$ is the graph whose vertices are the edges of $G$ and where the adjacency relation corresponds to the relation of having a common endpoint in $G$.

Actually, Chiarelli et al.~\cite{DBLP:journals/tcs/ChiarelliHJMP18} reported that the hardness of \MIF on line graphs was observed by Speckenmeyer in his PhD thesis~\cite{Speckenmeyer}. However, we were unable to find this result there. Therefore, we provide the easy proof, which boils down to essentially the same argument as in~\cite[Theorem 5]{DBLP:journals/tcs/ChiarelliHJMP18}.

\begin{lemma}\label{lem:fvs-clawfree}
\MIF is NP-complete on line graphs. Moreover, there is no algorithm solving \MIF in\/ $2^{o(n)}$ time for\/ $n$-vertex line graphs unless the ETH fails.
\end{lemma}

\begin{proof}
We reduce from the \textsc{Hamiltonian Path} problem, which is NP-complete and has no subexponential-time algorithm, even if the input graph has linearly many edges~\cite{platypus}. Let $G$ be a graph, which is the input instance of \textsc{Hamiltonian Path}.

First, note that any induced forest in $L(G)$
corresponds to a collection of vertex-disjoint paths in $G$. More formally, consider a set $E' \subseteq E(G)$, such that $L(G)[E']$ is a forest. We claim that the subgraph $G'=(V(G),E')$ of $G$ is a collection of vertex-disjoint paths.
Suppose not. This means that $G'$ contains a vertex $v$ of degree at least $3$ or a cycle $C$.
In the former case, the edges incident to $v$ in $G'$ form a clique in $L(G)[E']$. In the latter case, the edges of the cycle $C$ form a cycle in $L(G)[E']$. In either case, we get a contradiction to the assumption that $L(G)[E']$ is a forest.

We claim that $G$ has a Hamiltonian path if and only if $L(G)$ has an induced forest on $n-1$ vertices. Indeed, the $n-1$ edges of a Hamiltonian path in $G$ induce a path (in particular, a forest) in $L(G)$. For the converse, suppose that $L(G)$ has an induced forest on at least $n-1$ vertices. By the observation above, this induced forest corresponds to a collection of vertex-disjoint paths in $G$ with at least $n-1$ edges in total. This is only possible if this collection consists of a single path of length $n-1$, that is, a Hamiltonian path in $G$.

Finally, observe that the number of vertices of $L(G)$ is equal to the number of edges of $G$, which is linear in the number of vertices of $G$.
\end{proof}

Recall that line graphs are claw-free, that is, they contain no induced copy of $K_{1,3}$.
Thus \cref{lem:fvs-clawfree} implies that if $H$ contains any star with at least $3$ leaves, then \MIF remains NP-complete and has no subexponential-time algorithm on $H$-free graphs unless ETH fails. 
\cref{thm:fvs-neg} \ref{item:hardness} follows from combining \cref{lem:fvs-subdivide} and \cref{lem:fvs-clawfree}.

\newcommand{\bag}{\beta}
\newcommand{\cmp}{\alpha}

\section{Largest induced degenerate subgraph in low-treewidth graphs}\label{sec:deg-dp}

This section is devoted to the proof of the following result, which we used in \cref{sec:cors}.

\begin{lemma}\label{lem:deg-dp}
For every fixed\/ $d\in \N$, there is an algorithm for {\sc{Max Induced $d$-Degenerate Subgraph}} with running time\/ $2^{\Oh(w\log w)}\cdot n$, where\/
$w$ is the treewidth of the input graph and\/ $n$ is the number of its vertices.
\end{lemma}

\subparagraph*{Preliminaries on tree decompositions.}
First, we introduce some notation and terminology. 
A {\em{tree decomposition}} of a graph $G$ is a tree $T$ together with a mapping $\bag(\cdot)$ that assigns a \emph{bag} $\bag(x)$ to each node $x$ of $T$ in such a way that the following conditions hold:
\begin{enumerate}[label=(T\arabic*),ref=(T\arabic*),leftmargin=*]
\item\label{p:amoeba} for each $u\in V(G)$, the set of nodes $x$ with $u\in \bag(x)$ induces a connected non-empty subtree of $T$; and
\item\label{p:cover}  for each $uv\in E(G)$, there exists a node $x$ such that $\{u,v\}\subseteq \bag(x)$.
\end{enumerate}
The {\em{width}} of a tree decomposition $(T,\bag)$ is $\max_{x\in V(T)} |\bag(x)|-1$, and the {\em{treewidth}} of a graph $G$ is the minimum width of a tree decomposition of $G$.

Henceforth, all tree decompositions will be {\em{rooted}}: the underlying tree $T$ has a prescribed root vertex~$r$. This gives rise a natural ancestor-descendant relation: 
we write $x\preceq y$ if $x$ is an ancestor of $y$ (where possibly $x=y$). Then, for a node $x$ of $T$, we define the {\em{component}} at $x$ as
\[\cmp(x)=\biggl(\bigcup_{y\succeq x} \bag(y)\biggr)\setminus \bag(x).\]
It easily follows from~\ref{p:amoeba} and~\ref{p:cover} that then $N(\cmp(x))\subseteq \bag(x)$ for every node $x$.

A {\em{nice tree decomposition}} is a normalized form of a rooted tree decomposition in which every node is of one of the following four kinds.
\begin{itemize}
\item {\bf{Leaf node}}: a node $x$ with no children and with $\bag(x)=\emptyset$.
\item {\bf{Introduce node}}: a node $x$ with one child $y$ such that $\bag(x)=\bag(y)\cup \{u\}$ for some vertex $u\notin \bag(y)$.
\item {\bf{Forget node}}: a node $x$ with one child $y$ such that $\bag(x)=\bag(y)\setminus \{u\}$ for some vertex $u\in \bag(y)$.
\item {\bf{Join node}}: a node $x$ with two children $y$ and $z$ such that $\bag(x)=\bag(y)=\bag(z)$.
\end{itemize}
Moreover, we require that the root $r$ of the nice tree decomposition satisfies $\bag(r)=\emptyset$.

It is known that any given tree decomposition $(T,\bag)$ of width $k$ of an $n$-vertex graph $G$ can be transformed in $k^{\Oh(1)}\cdot \max(n,|V(T)|)$ time into a nice tree decomposition of $G$ of width at most as large, 
see~\cite[Lemma 7.4]{platypus}. Moreover, given an $n$-vertex graph $G$ of treewidth $w$, a tree decomposition of $G$ of width at most $5w+4$ can be computed in $2^{\Oh(w)}\cdot n$ time~\cite{BodlaenderDDFLP16}, 
and this tree decomposition has at most $n$ nodes. By combining these two results, for the proof of \cref{lem:deg-dp}, we can assume that the input graph $G$ is supplied with a nice tree decomposition $(T,\bag)$ 
of width $k\leq 5w+4$, where $w=\tw(G)$.
From now on, our goal is to design a suitable dynamic programming algorithm working on this decomposition with running time $2^{\Oh(k\log k)}\cdot n=2^{\Oh(w\log w)}\cdot n$.

\subparagraph*{Dynamic programming states.}
The main idea behind our dynamic programming algorithm is to view the notion of degeneracy via vertex orderings, as expressed in the following fact.

\begin{lemma}[Folklore]
A graph\/ $H$ is\/ $d$-degenerate if and only if there is a linear ordering\/ $\sigma$ of vertices of\/ $H$ such that every vertex of\/ $H$ has at most\/ $d$ neighbors that are smaller in\/ $\sigma$.
\end{lemma}

Hence, the problem considered in \cref{lem:deg-dp} can be restated as follows: find a largest set $A\subseteq V(G)$ that admits a linear ordering $\sigma$ 
in which every vertex of $A$ has at most $d$ neighbors in $G[A]$ that are smaller in $\sigma$. Intuitively, our dynamic programming will therefore keep track of the intersection of the bag with $A$, 
the restriction of $\sigma$ to this intersection; and how many smaller neighbors of each vertex from this intersection have been already forgotten.

We now proceed with formal details. For a node $x$ of $T$, a set $X\subseteq \bag(x)$, 
a linear ordering $\sigma$ of $X$, and a function $f\colon X\to \{0,\ldots,d\}$, we define
$\Phi_x[X,\sigma,f]\in \N$ as follows. The value $\Phi_x[X,\sigma,f]$ is the maximum size of a set $Y\subseteq \cmp(x)$ such that $X\cup Y$ admits a linear ordering 
$\tau$ with the following properties: $\tau$ restricted to $X$ is equal to $\sigma$ and for every $a\in X$, there are at most $f(a)$ vertices $b\in Y$ that are adjacent to $a$ and smaller than $a$ in $\tau$.
Note that other neighbors of $a$ that belong to $X$ are \emph{not} taken into consideration when verifying the quota imposed by $f(a)$.
Note also that such a set $Y$ always exists, as $Y=\emptyset$ satisfies the criteria.

For a fixed node $x$, the total number of triples $(X,\sigma,f)$ as above is at most
\[2^{k+1}\cdot (k+1)!\cdot (d+1)^{k+1}\leq 2^{\Oh(k\log k)}.\]
Hence, we now show how to compute the values $\Phi_x[X,\sigma,f]$ in a bottom-up manner, so that the values for a node $x$ are computed based on the values for the children of $x$ in $2^{\Oh(k\log k)}$ time.
The answer to the problem corresponds to the value $\Phi_r[\emptyset,\emptyset,\emptyset]$, where $r$ is the root of~$T$. 
While $\Phi_r[\emptyset,\emptyset,\emptyset]$ is just the size of a largest feasible solution, an actual solution can be recovered 
from the dynamic programming tables using standard methods within the same complexity: for every computed value $\Phi_x[X,\sigma,f]$, we store the way this value was obtained, 
and then we trace back the solution from $\Phi_r[\emptyset,\emptyset,\emptyset]$ in a top-down manner.

\subparagraph*{Transitions.}
It remains to provide recursive formulas for the values of $\Phi_x[\cdot,\cdot,\cdot]$. We only present the formulas, while the verification of their correctness, which follows easily from the definition of $\Phi_x[\cdot,\cdot,\cdot]$, is left to the reader. As usual, we distinguish cases depending on the type of $x$.
\begin{itemize}
\item {\bf{Leaf node}} $x$. Then we have only one value:
\[\Phi_x[\emptyset,\emptyset,\emptyset]=0.\]
\item {\bf{Introduce node}} $x$ with child $y$ such that $\bag(x)=\bag(y)\cup \{u\}$.
Then
\[\Phi_x[X,\sigma,f]=\begin{cases}\Phi_y[X,\sigma,f] & \textrm{if }u\notin X;\\ \Phi_y[X\setminus \{u\},\sigma|_{X\setminus \{u\}},f|_{X\setminus \{u\}}] & \textrm{if }u\in X.\end{cases}\]
\item {\bf{Forget node}} $x$ with child $y$ such that $\bag(x)=\bag(y)\setminus \{u\}$.
Then we have
\[\Phi_x[X,\sigma,f] = \max\,\left(\, \Phi_y[X,\sigma,f],\ 1 + \max_{(\sigma',f')\in S(X,\sigma,f)} \Phi_y[X\cup \{u\},\sigma',f']\,\right),\]
where $S(X,\sigma,f)$ is the set comprising the pairs $(\sigma',f')$ satisfying the following:
\begin{itemize}
\item $\sigma'$ is a vertex ordering of $X\cup \{u\}$ whose restriction to $X$ is equal to $\sigma$; and
\item $f'\colon X\cup \{u\}\to \{0,\ldots,d\}$ is such that for all $a\in X$ that are adjacent to $u$ and larger than $u$ in $\sigma'$, we have $f'(a)\leq f(a)-1$, and for all other $a\in X$, we have $f'(a)\leq f(a)$.
Moreover, we require that $f'(u)\leq d-\ell$, where $\ell$ is the number of vertices $a\in X$ that are adjacent to $u$ and smaller than $u$ in $\sigma'$.
\end{itemize}
\item {\bf{Join node}} $x$ with children $y$ and $z$. Then
\[\Phi_x[X,\sigma,f] = \max_{f_y+f_z\leq f} \Phi_y[X,\sigma,f_y]+\Phi_z[X,\sigma,f_z],\]
where $f_y+f_z\leq f$ means that $f_y(a)+f_z(a)\leq f(a)$ for each $a\in X$.
\end{itemize}
It is straightforward to see that using the formulas above, each value $\Phi_x[X,\sigma,f]$ can be computed in $2^{\Oh(k\log k)}$ time based on the values computed for the children of $x$.
This completes the proof of \cref{lem:deg-dp}.


\end{document}